\documentclass{svproc}
\usepackage{xcolor}
\usepackage{url}
\usepackage{graphicx}
\usepackage{soul}
\usepackage{booktabs}
\usepackage{amsmath}
\usepackage{subcaption}
\usepackage[ruled,linesnumbered,boxed]{algorithm2e}
\usepackage{amssymb}
\SetKwComment{Comment}{$\triangleright$\ }{}
\SetKwInOut{Parameter}{Parameters}
\usepackage{mathptmx}       
\usepackage{helvet}         
\usepackage{courier}        
\usepackage{type1cm}        
\usepackage{makeidx}         
\usepackage{graphicx}        
\usepackage{multicol}        
\usepackage[bottom]{footmisc}

\begin{document}
\mainmatter              
\title{Patterns of Multiplex Layer Entanglement across Real and Synthetic Networks}
\titlerunning{Entanglement in multiplex networks}  
%
\author{Bla\v{z} \v{S}krlj$^{1}$ \and Benjamin Renoust$^{2}$}
\authorrunning{\v{S}krlj and Renoust} 
%
\tocauthor{Bla\v{z} \v{S}krlj and Benjamin Renoust}
\institute{Jo\v{z}ef Stefan Institute, Jamova 38, Ljubljana,\\
Jo\v{z}ef Stefan International Postgraduate School, Jamova 38, Ljubljana\\
\email{blaz.skrlj@ijs.si},\\
\and
Osaka University, Institute for Datability Science\\
\email{renoust@ids.osaka-u.ac.jp}\\
}

\maketitle              

\begin{abstract}
Real world complex networks often exhibit multiplex structure, connecting entities from different aspects of physical systems such as social, transportation and biological networks. Little is known about general properties of such networks across disciplines. In this work, we first investigate how consistent are connectivity patterns across 35 real world multiplex networks. We demonstrate that entanglement homogeneity and intensity, two measures of layer consistency, indicate apparent differences between social and biological networks. We also investigate trade, co-authorship and transport networks. We show that real networks can be separated in the joint space of homogeneity and intensity, demonstrating the usefulness of the two measures for categorization of real multiplex networks. Finally, we design a multiplex network generator, where similar patterns (as observed in real networks), are emerging over the analysis of 11{,}905 synthetic multiplex networks with various topological properties.

\end{abstract}

\section{Introduction}
\label{sec:introduction}

Real-world networks commonly consist of different types of entities, all connected into a single system. The abstraction of \emph{multiplex} networks offers a structure, capable of capturing the key parts of such systems, such as connectivity patterns. Multiplex networks emerge, and were studied in biology, social sciences, finance, logistics and more. They are both theoretically interesting, as well as practically useful \cite{battiston2014structural}. Recently, the notions of multiplex community detection and centralities have been a lively research area, indicating many insights can be obtained by studying such rich structures directly, without simplification \cite{gomez2013diffusion,vskrlj2019cbssd,chen2018suppressing} (\textit{e.g.}, aggregation into a single node type). Multiplex networks offer the opportunity to simultaneously explore multiple aspects of the same system~\cite{kivela2019visual}, and are as such indispensable for the study of \textit{e.g.}, biological or social networks, where entities can be naturally observed with respect to different aspects \textit{(e.g.}, an user on Twitter, Facebook and Snapchat is the same physical person, yet can be studied with respect to individual social networks where it is present).

The ideas, that influenced this work the most are discussed next. 
Since the structure of a multilayer corresponds to its layers and aspects~\cite{kivela2014multilayer}, the analysis of the organization of layers is key to understanding the properties of a multiplex network~\cite{renoust2015detangler}. The analysis of the overlapping edges between layers, namely \textit{edge entanglement}~\cite{renoust2014entanglement} studies how the different layers of a multiplex network intertwine to form a coherent whole. 

Even though the ideas related to description of multiplex networks are being actively developed \cite{wang2018social,omodei2015characterizing}, we believe little effort is focused on evaluation of such measures at larger scales, across multiple disciplines and contexts. This work was also inspired by multilayer flow analysis \cite{de2015identifying}, where distinct structures, describing parts of networks emerged.
The contributions of this work are multiple, and are described next:

\begin{itemize}
\item We present an efficient implementation of multiplex homogeneity and intensity, the two measures used in this work \cite{renoust2014entanglement}
\item Both measures, along with normalized homogeneity, are computed for the first time on 35 real-world multiplex networks.
\item We demonstrate a distinct relationship between homogeneity and intensity, showing the two measures can separate between different types of multiplex networks. 
\item We present a multiplex network generator that produced networks with various degrees of intensity and homogeneity. We generated  11{,}905 synthetic networks, where patterns, similar to the ones in the real networks emerged.
\end{itemize}

\section{Multiplex networks}
A multiplex network can be defined as a sequence $M = \{G_l\}_{l \in L} = \{(V_{l},E_{l})\}_{l \in L}$ where $E_{l} \subseteq N_l \times V_l$ is a set of edges in one network $l \in L$ of the sequence~\cite{kivela2014multilayer}. 
Multiplex networks are commonly understood as layers comprised of interactions, where each layer corresponds to a specific aspect of the system, and nodes represent \emph{the same} entity across all layers. We represent a multiplex network as a structure $M=(V_M, E_M)$, where $V_M$ is the set of nodes and $E_M$ the set of all edges (in all layers). 

For example, a biological system can be studied at the protein, RNA or gene level~\cite{valdeolivas2018random}, and similarly, social networks can be studied by taking into account a person's presence on multiple platforms~\cite{mittal2019analysis}.
For computational purposes, such networks are commonly represented in the form of supra-adjacency matrices, where block-diagonal structure, connecting the same node across individual layers emerges \cite{cozzo2015structure}. Algorithms can operate on such matrices directly and thus exploit such additional information representing multiple aspects. Such approaches are useful when node-level information is considered.

Algorithms for analysis of multiplex networks can also operate on sparse, adjacency structure of the multiplex network directly, yet need to take into account that a given node is present in multiple layers. Such representation is suitable for this work, as we are focused primarily on how edges co-occur across \emph{layers}. Hence, this work focuses primarily on the relations between \emph{the layers} of a given multiplex network.
We next discuss the two measures we consider throughout this work.

\section{Multiplex entanglement and intensity}
\label{sec:entanglement}

We briefly discuss the entanglement measures definitions from previous work~\cite{renoust2014entanglement}.

\subsection{Layer interaction network}

Recall our multiplex network $M=(V_M,E_M)=\{G_l \}_{l \in L}$. Such a network really distinguishes itself from classical graphs through the use of different layers to connect nodes. These layers may have different patterns and may overlap together. There may even exist latent dependencies among these layers.
To investigate this matter, each layer could be abstracted to one single node and form a new graph, the Layer Interaction Network (\textit{LIN})~\cite{renoust2014entanglement}. Visualizing the LIN is a key component for multiplex network visualization such as in Detangler~\cite{renoust2015detangler}.

In the LIN, $LIN=(L,F)$, each node corresponds to a layer $l, l', l'', \ldots \in L$ of the multilayer network $M$, and each edge $f \in F$ captures when two layers overlap through edges. More formally, there exist an edge $f=(l, l')$ whenever there exists at least two nodes $u, v \in V_M$ such that there exists at least one edge connecting these two nodes on each layer $e_M = (u,v) \in l$ and $e'_M = (u,v) \in l'$. 
The LIN can be interpreted as an edge-layer co-occurrence graph, and the weight of an edge $f = (l, l')$, denoted as $n_{l, l'}$ equals the number of times layers $l, l'$ co-occur. 
By extension, $n_{l, l}$ is the number of edges on layer $l$. This process is illustrated in Figure~\ref{fig:toyexample}.

\begin{figure}[t]
\centering
\includegraphics[width=\linewidth]{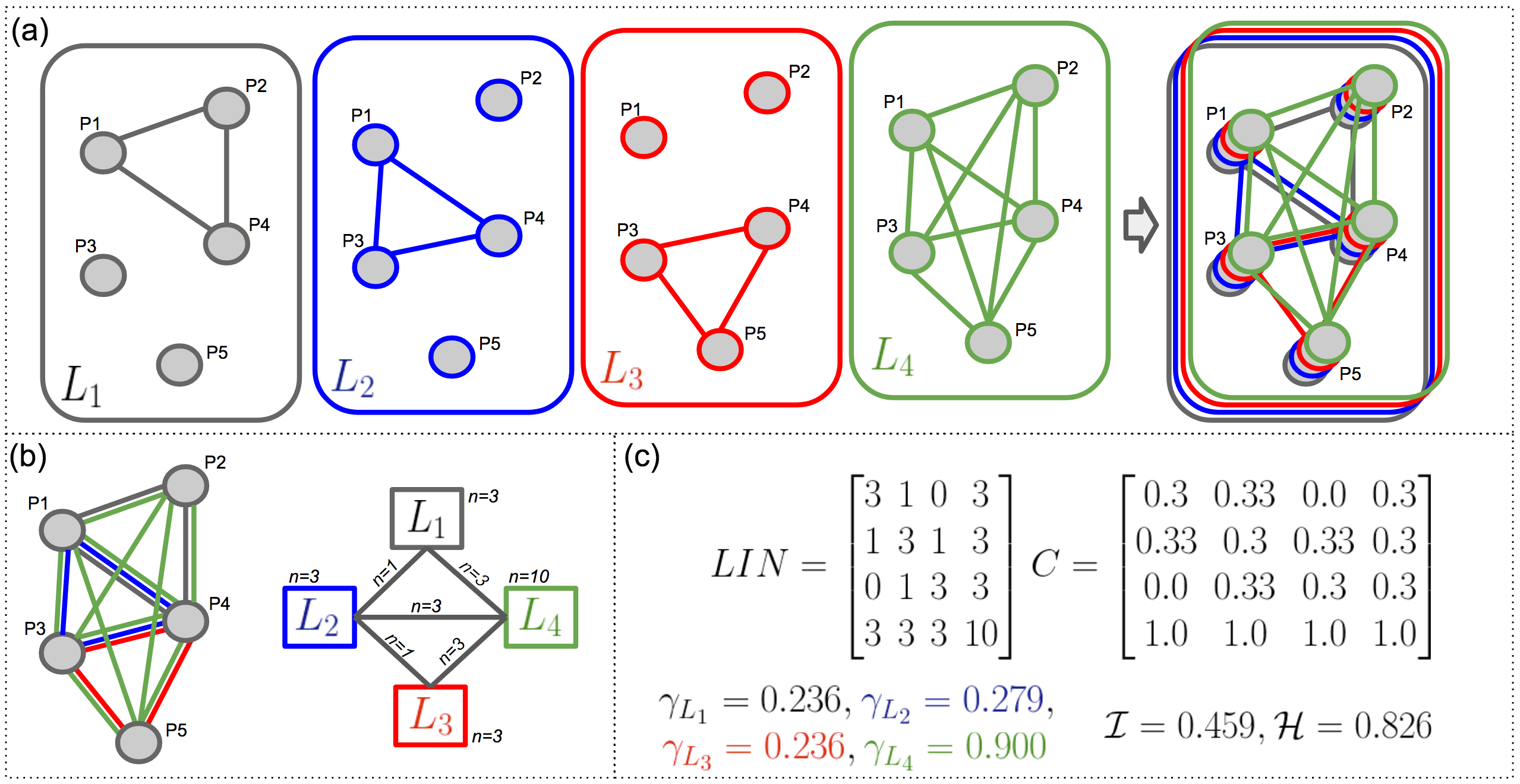}
\caption{A toy example of layer entanglement computation: a) separated layers considered in a multilayer network; b) constructing the layer interaction network from the example; c) measuring entanglement from the example.}
\label{fig:toyexample}
\end{figure}

\subsection{Layer entanglement}
\label{subsec:entanglement}

The analysis of edge entanglement is inspired by the analysis of relation content in social networks~\cite{burt1985relation}. The idea is to study the redundancy between relation content, each forming in our formalism a different layer. The edge entanglement measures the ``influence'' of a layer in its neighborhood.

This measure is recursively defined: the entanglement $\gamma_l$ of a layer $l$ is defined upon the entanglement of the layers it is entangled with. Similarly to the eigen centrality~\cite{Wasserman1994}, this translates into the recursive equation: $$\gamma_l.\lambda=\sum_{l'\in T}{\frac{n_{ll'}}{n_l}\gamma_l'}.$$ The entanglement of a layer $\gamma_l$ can be retrieved from a vector $\vec{\gamma}$ which corresponds to the right eigenvector (associated to the maximum eigenvalue $\lambda$) of the layer overlap frequency matrix with corresponding overlap, defined as:
\begin{equation*}
C = (c_{ll'}), \quad \textrm{where} \quad c_{ll'} = \frac{n_{ll'}}{n_{ll}}.
\end{equation*}
this metric was initially discussed in \cite{renoust2014entanglement}), and is constructed using the weights in the $LIN$ (see Figs.~\ref{fig:toyexample} and~\ref{fig:maxHI}).

\subsection{Entanglement intensity and homogeneity}

\begin{figure}[t]
\centering
\includegraphics[width=\linewidth]{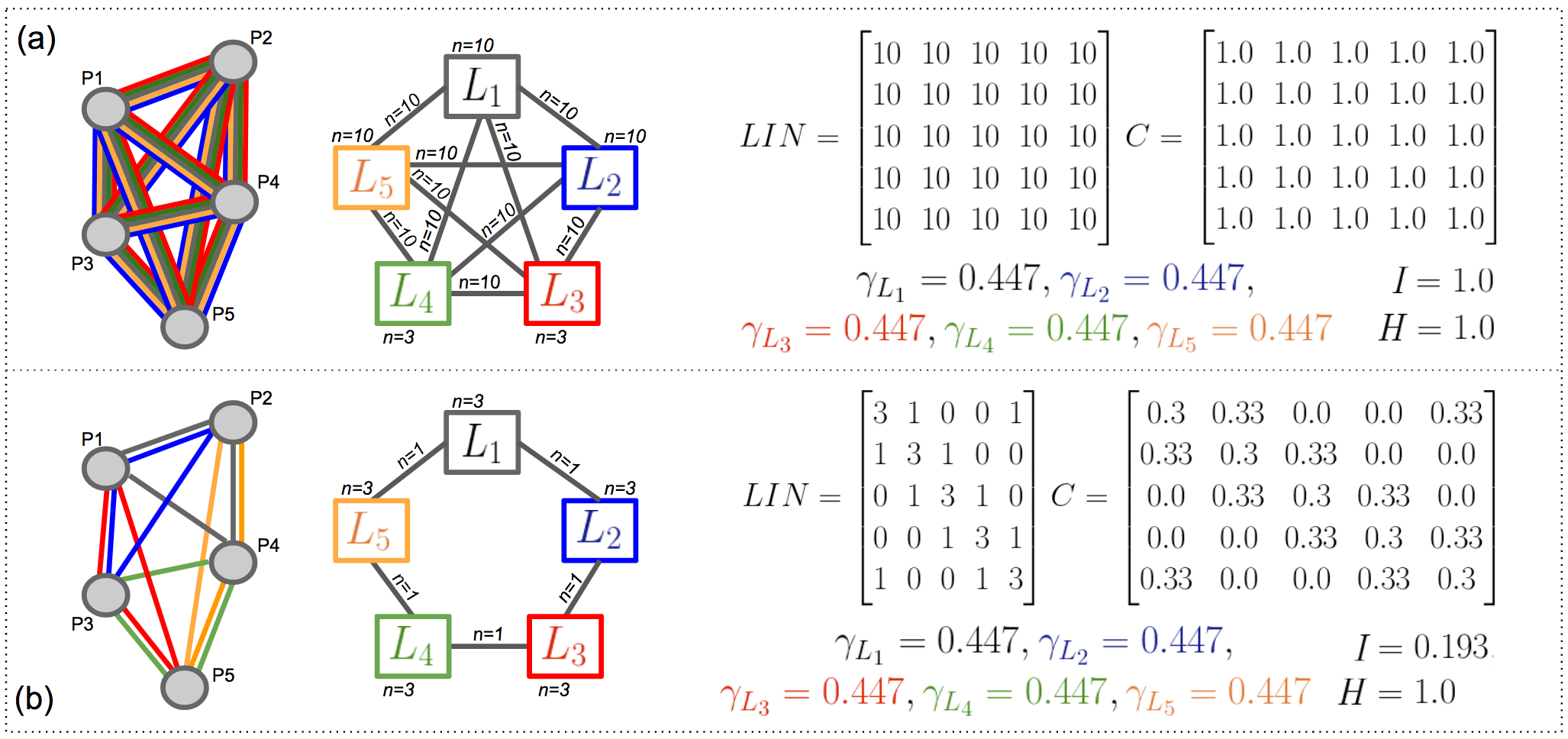}
\caption{Two very different cases of maximum homogeneity $\mathit{H}=1$, the multiplex network and the $LIN$ are shown, with matrices and entanglement measures. a) all layers are saturating all edges, so we have maximum intensity $\mathit{I}=1$; b) layers are well balanced, but we may have a lot more interactions possible.}
\label{fig:maxHI}
\end{figure}

The layer entanglement $\gamma_l$ measures the share of layer $l$ overlapping with other layers, so that nodes of $M$ are connected. The more a group of layers interacts together, the more the nodes they connect will be cohesive in view of \textit{these} layers, hence the more $\gamma_l\; \forall l \in L $ values will be similar (their share of entanglement will be similar). This is captured by the \textit{entanglement homogeneity}~\cite{renoust2014entanglement} which is then defined as the following cosine similarity: $$\mathit{H}=\frac{<1_L,\gamma>}{\lVert1_L \rVert \lVert\gamma \rVert} \in [0,1].$$ Optimal homogenity is not necessarily reached only when all nodes are connected through all layers, but also when all nodes are connected in a very balanced manner between all layers (see Figure~\ref{fig:maxHI}). Homogeneity thus permits various \emph{symmetries} in a given LIN.

When a maximum overlap is reached through all layers in the network, the frequencies in the matrix $C$ (of size $|L| \times |L|)$ are saturated with $C_{i,j}=1$. This gives us a theoretical limit to measure the amount of layer overlap through the \textit{entanglement intensity}~\cite{renoust2014entanglement}, defined as: $$\mathit{I}=\lambda/|L|.$$

\section{A multiplex network generator}
\label{sec:multiplex-generator}
In this section, we describe an algorithm for generation of multiplex networks based on the following observations. Let $M = (V_M,E_M)$  represent a multiplex network with $L$ layers. 
Each node is associated to a random number of layers $\{l_1, l_2, \ldots, l_i\} \subseteq L$.
Now for each layer $l_i \in L$ there is a set of nodes $V_{l_i} \subseteq V_M$ which form a potential set of edges of size $|E_{l_i}| = \frac{1}{2}|V_{l_i}|(|V_{l_i}|-1)$. We introduce the probability $p$ of an edge to be created between any pair of node on a layer so we may avoid cliques to form on each layer. For algorithmic reasons, this probability is implemented as an edge dropout $d$ such as $p = 1-d$,  and randomly prune edges from a potential clique. Thus, the higher the $d$, the sparser the network. Intuitively, the more similar a given random multiplex is to a clique over each layer, the higher its intensity. 
\begin{algorithm}[t]
\Parameter{Number of nodes $v$, number of layers $k$, dropout $d$}
\KwResult{A multiplex network $M$}
$M$ $\leftarrow$ emptyMultiplexObject\;
\For{for node in $[1 \dots v]$}{
    numberOfLayers $\leftarrow$ randomNumber($k$) \Comment*[r]{Layer presence is random.}
    layerNodes $\leftarrow$ assignNodeToLayers(node, numberOfLayers)\;
    update($M$, layerNodes)\Comment*[r]{Update global network.}
}
\For{layer $l_i$ with corresponding node set $V_{l_i}$}{
    nodeClique $\leftarrow$ generator of node pairs from $V_{l_i}$\;
    finalLayer $\leftarrow$ sampleWithProbability(nodeClique, $1-d$)\Comment*[r]{Sample via $d$.}
    update($M$,finalLayer)\Comment*[r]{Update global network.}
}
\Return{$M$}\;
\caption{Multiplex network generator.}
 \label{algo:rep}
\end{algorithm}
The purpose of this generator is to offer a simple testbed for further exploration, as well as additional evidence of the relation between homogeneity and intensity on many random, synthetic networks. The Algorithm~\ref{algo:rep} represents the proposed procedure.

The generator first randomly assigns the same node index to randomly many layers (lines 1-6). Once assigned, the layers are processed by applying the dropout on $|V_{l_i}| \choose 2$ possible edges in layer $l_i$. 
The global multiplex is updated during this process (lines 7- 12). Note that in line 8, the whole clique is virtually generated. This step is not necessary, as commonly only a small number of edges need to be sampled from all possible edge combinations. The implementation thus uses a generator with lazy evaluation, avoiding potential combinatorial explosion with a large number of nodes (very large networks).

\subsection{Some theoretical properties of the generator}
In this section we show two properties of the proposed generator. We denote $v = |V_M|$ the parameter setting the number of nodes of the network, $k = |L|$ the parameter setting the number of edge layers in the network, and $d$ the edge dropout. 
\begin{proposition}[Number of edges]
Let $\phi \in \mathbb{N}^{+}$ represent the number of possible edges. Then $\phi \leq k \cdot \binom{v}{2}.$
\end{proposition}
\begin{proof}
Note that in multiplex layers, each layer can have at most $v$ nodes. Assuming they form a clique, each layer is thus comprised of $v \choose 2$ nodes. As there are $k$ layers, there can be at most $k \cdot \binom{v}{2}$ edges --- a clique of $v$ nodes in each layer. As each layer is during generation subject to dropout, which is neglected, when set to 0 (no edges are erased), we refer to this bound as $\phi \leq k \cdot \binom{v}{2}$.  \qed
\end{proof}

\begin{corollary}[Time complexity]
In lower limit, $d \rightarrow 0$, thus a full clique needs to be constructed, assuming each node is projected across all layers. The complexity \textit{w.r.t.} the number of layers and edges is: $\mathcal{O}(k \cdot \binom{v}{2}) = \mathcal{O}(|E_M|).$
\end{corollary}

Note that even though, theoretically, the proposed generator generates a clique and then samples from it, current, lazy implementation only \emph{generates} the edges needed to satisfy a given $d$ percentage. In practice, only when $d \approx 0$, the generator needs larger portions of space (and time). As such, fully connected networks do not represent real systems, we were able to generate networks with tens of thousands of nodes using this approach.

\section{Empirical evaluation}

\begin{table}[h]
\centering
\caption{Real multiplex networks and their properties. The ID in the second column corresponds to Figure 3 (c).}
\resizebox{0.8\textwidth}{!}{
\begin{tabular}{lccccccc}\toprule
                   Dataset &    ID &      Type &    Nodes &    Edges &  Number of layers &  Mean degree &      CC \\ \midrule
          arXiv-Netscience~\cite{de2015identifying} & 6 &Coauthorship &    26796 &    59026 &                13 &         4.41 &    3660 \\
               PierreAuger~\cite{de2015identifying} & 22 &  Coauthorship &      965 &     7153 &                16 &        14.82 &     131 \\
               Arabidopsis~\cite{stark2006biogrid} &  39  &   Genetic &     8765 &    18655 &                 7 &         4.26 &     387 \\
                       Bos~\cite{stark2006biogrid} &  3   &  Genetic &      369 &      322 &                 4 &         1.75 &      82 \\
                   Candida~\cite{stark2006biogrid} &    23,24,25   & Genetic &      418 &      398 &                 7 &         1.90 &      50 \\
                  Celegans~\cite{stark2006biogrid} &    32  & Genetic &     4557 &     8182 &                 6 &         3.59 &     193 \\
                DanioRerio~\cite{stark2006biogrid} & 26,27    &   Genetic &      180 &      188 &                 5 &         2.09 &      45 \\
                Drosophila~\cite{stark2006biogrid} & 31    &   Genetic &    11970 &    43367 &                 7 &         7.25 &     346 \\ 
                    Gallus~\cite{stark2006biogrid} &  16 &    Genetic &      367 &      389 &                 6 &         2.12 &      54 \\
           HepatitusCVirus~\cite{stark2006biogrid} &   33  &  Genetic &      129 &      137 &                 3 &         2.12 &       4 \\
                Homo Sapiens~\cite{stark2006biogrid} &   30  &  Genetic &    36194 &   170899 &                 7 &         9.44 &     785 \\
              HumanHerpes4~\cite{stark2006biogrid} &    29 &  Genetic &      261 &      259 &                 4 &         1.98 &      21 \\
                 HumanHIV1~\cite{stark2006biogrid} &   5   & Genetic &     1195 &     1355 &                 5 &         2.27 &      13 \\
               Oryctolagus~\cite{stark2006biogrid} &   7  &  Genetic &      151 &      144 &                 3 &         1.91 &      21 \\
                Plasmodium~\cite{stark2006biogrid} &  9  &    Genetic &     1206 &     2522 &                 3 &         4.18 &      27 \\
                    Rattus~\cite{stark2006biogrid} &    40 &   Genetic &     3263 &     4268 &                 6 &         2.62 &     296 \\
                 SacchCere~\cite{stark2006biogrid} &    2  & Genetic &    27994 &   282755 &                 7 &        20.20 &     432 \\
                 SacchPomb~\cite{stark2006biogrid} &1    &   Genetic &    10178 &    63677 &                 7 &        12.51 &     286 \\
                   Xenopus~\cite{stark2006biogrid} &  37, 38 &    Genetic &      582 &      620 &                 5 &         2.13 &     109 \\
            YeastLandscape~\cite{costanzo2010genetic} &   34 &   Genetic &    17770 &  8473997 &                 4 &       953.74 &       4 \\
                  CElegans~\cite{chen2006wiring} &   20  &  Neuronal &      791 &     5863 &                 3 &        14.82 &       6 \\
                Cannes2013~\cite{omodei2015characterizing} &  8  &  Social &   659951 &   991854 &                 3 &         3.01 &   48375 \\
 CKM-Physicians-Innovation~\cite{coleman1957diffusion} &  19   &  Social &      674 &     1551 &                 3 &         4.60 &      12 \\
                CS-Aarhus~\cite{magnani2013combinatorial} &   36  &   Social &      224 &      620 &                 5 &         5.54 &      13 \\
      Kapferer-Tailor-Shop~\cite{kapferer1972strategy} &  35   &   Social &      150 &     1018 &                 4 &        13.57 &       5 \\
      Krackhardt-High-Tech~\cite{krackhardt1987cognitive} &  13  &    Social &       63 &      312 &                 3 &         9.90 &       3 \\
           Lazega-Law-Firm~\cite{lazega2001collegial} &  18   &   Social &      211 &     2571 &                 3 &        24.37 &       3 \\
                MLKing2013~\cite{omodei2015characterizing} &   14  &   Social &   392542 &   396671 &                 3 &         2.02 &   36041 \\
       MoscowAthletics2013~\cite{omodei2015characterizing} &   17 &     Social &   133619 &   210250 &                 3 &         3.15 &    6323 \\
         ObamaInIsrael2013~\cite{omodei2015characterizing} &   21   &  Social &  3457453 &  4061960 &                 3 &         2.35 &  651141 \\
 Padgett-Florence-Families~\cite{padgett1993robust} &   28  &   Social &       26 &       35 &                 2 &         2.69 &       2 \\
   Vickers-Chan-7thGraders~\cite{vickers1981representing} &  0   &   Social &       87 &      740 &                 3 &        17.01 &       3 \\
                       FAO~\cite{de2015structural} & 15  &   Trade &    41713 &   318346 &               364 &        15.26 &     571 \\
                     EUAir~\cite{cardillo2013emergence} &  4 & Transport &     2034 &     3588 &                37 &         3.53 &      41 \\
                    London~\cite{de2014navigability} &  11,12  &  Transport &      399 &      441 &                 3 &         2.21 &       3 \\
                \bottomrule
\end{tabular}
}
\label{tab:summary}
\end{table}


In this section we discuss the empirical evaluation of the two considered measures across a series of real world networks. 

All considered networks are summarized in Table~\ref{tab:summary}\footnote{The networks are hosted at \url{https://comunelab.fbk.eu/data.php}}. All considered networks are static. We computed the metrics for all connected components.
The entanglement algorithm was integrated into the Py3plex library \cite{vskrlj2018py3plex}\footnote{\url{https://github.com/SkBlaz/Py3plex/blob/master/examples/example_entanglement.py}. The generator is accessible at \url{https://github.com/SkBlaz/Py3plex/blob/master/py3plex/core/random_generators.py}}. For each network, we computed homogeneity and intensity. For the generation of synthetic networks, we used the following hyperparameter ranges: 
\begin{itemize}
\item $v \in \{10,25,50,100,250,500,1000,2500\}$
\item $k \in \{3,4,5,6,7,8,9,10\}$
\item $d \in \{0.001,0.9,0.01\}$
\end{itemize}

\section{Results}

In this section we present the results of empirical evaluation. For readability purposes, we visualize individual results as distributions of a given score across network types. We first show entanglement metrics on real networks in Figure~\ref{fig:real}. We next present the results on the generated networks in Figure~\ref{fig:syn}.

\begin{figure}[h!]
\centering
\begin{tabular}{cc}
\subcaptionbox{Real networks: homogeneity $H$}{\includegraphics[width = 2.3in]{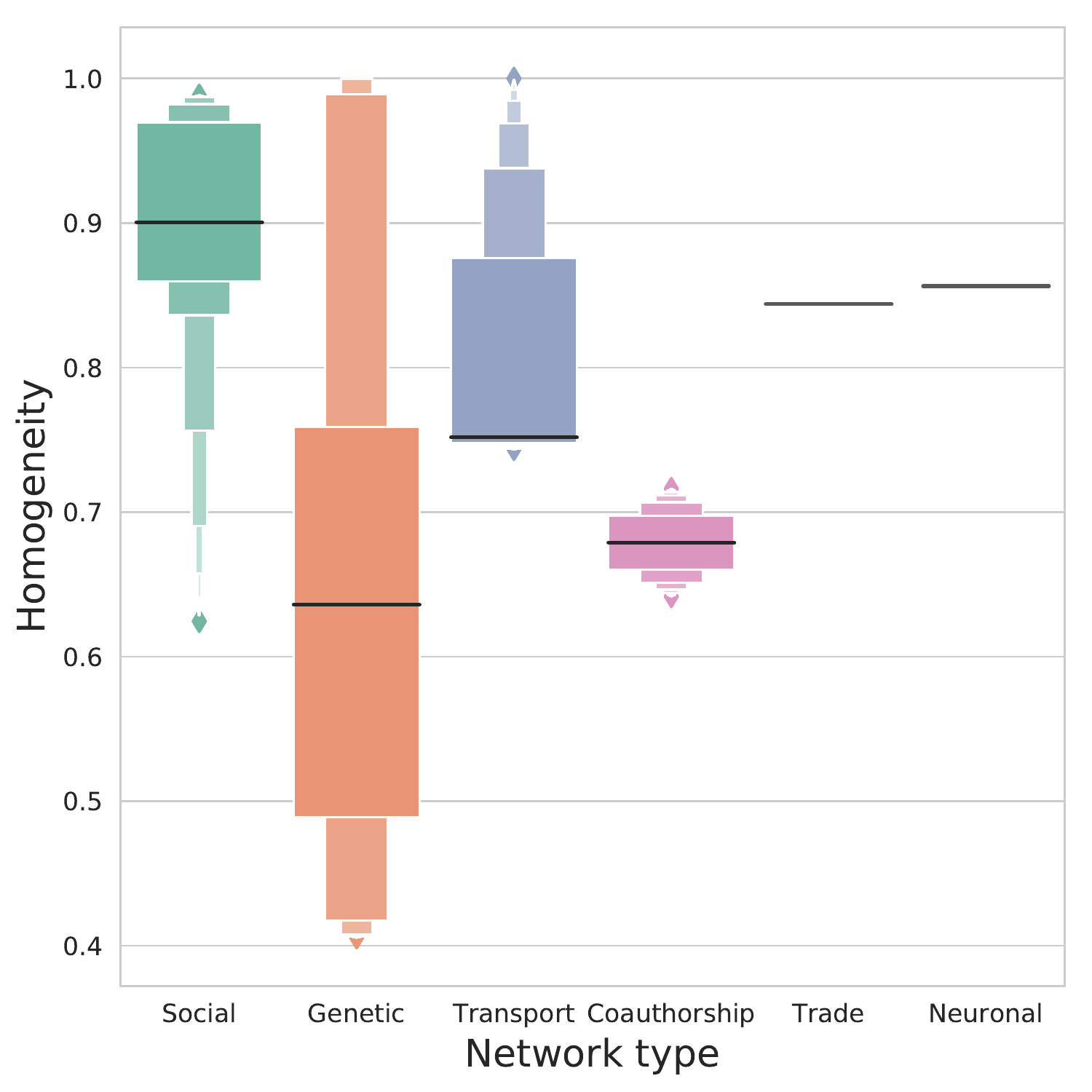}} &
\subcaptionbox{Real networks: intensity $I$}{\includegraphics[width = 2.3in]{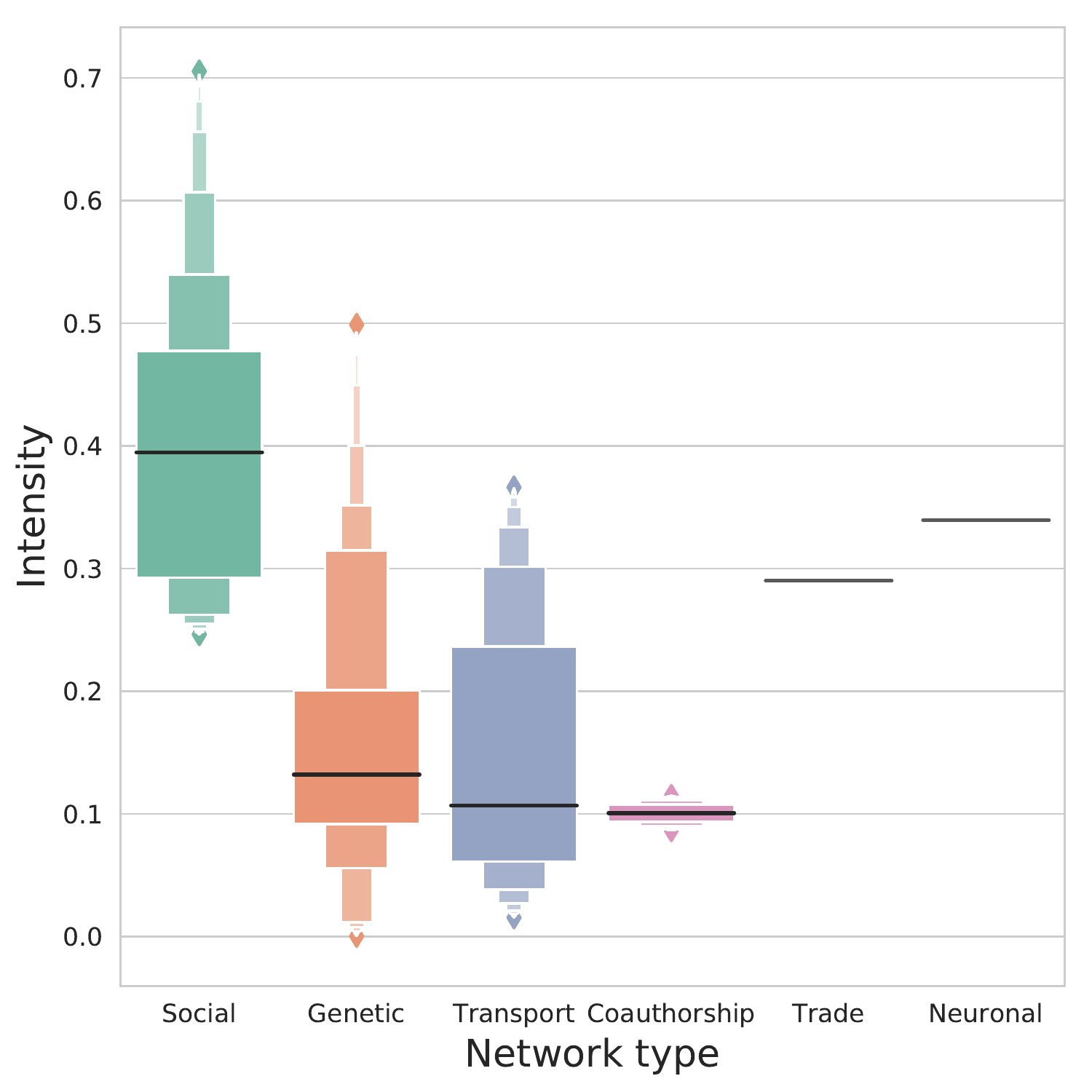}} \\
\end{tabular}
\subcaptionbox{Real networks: $H \times I$}{\includegraphics[width = 0.7\linewidth]{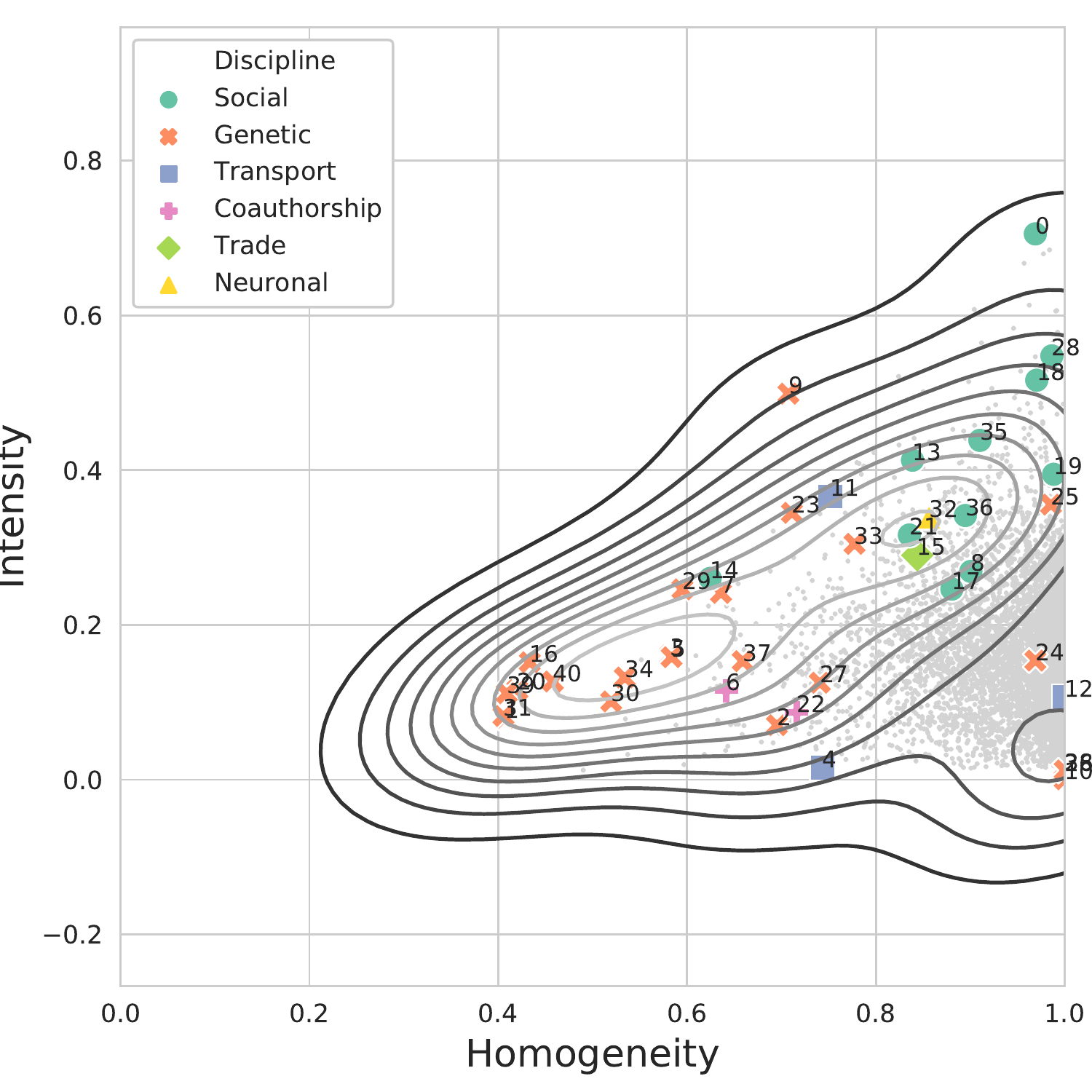}}
\caption{Results on real networks. Labels in (c) map to Table~\ref{tab:summary} (ID). Gray dots represent synthetic samples.}
\label{fig:real}
\end{figure}

\begin{figure}[h!]
\centering
\begin{tabular}{cc}
\subcaptionbox{Synthetic networks: homogeneity $H$}{\includegraphics[width = 2.2in]{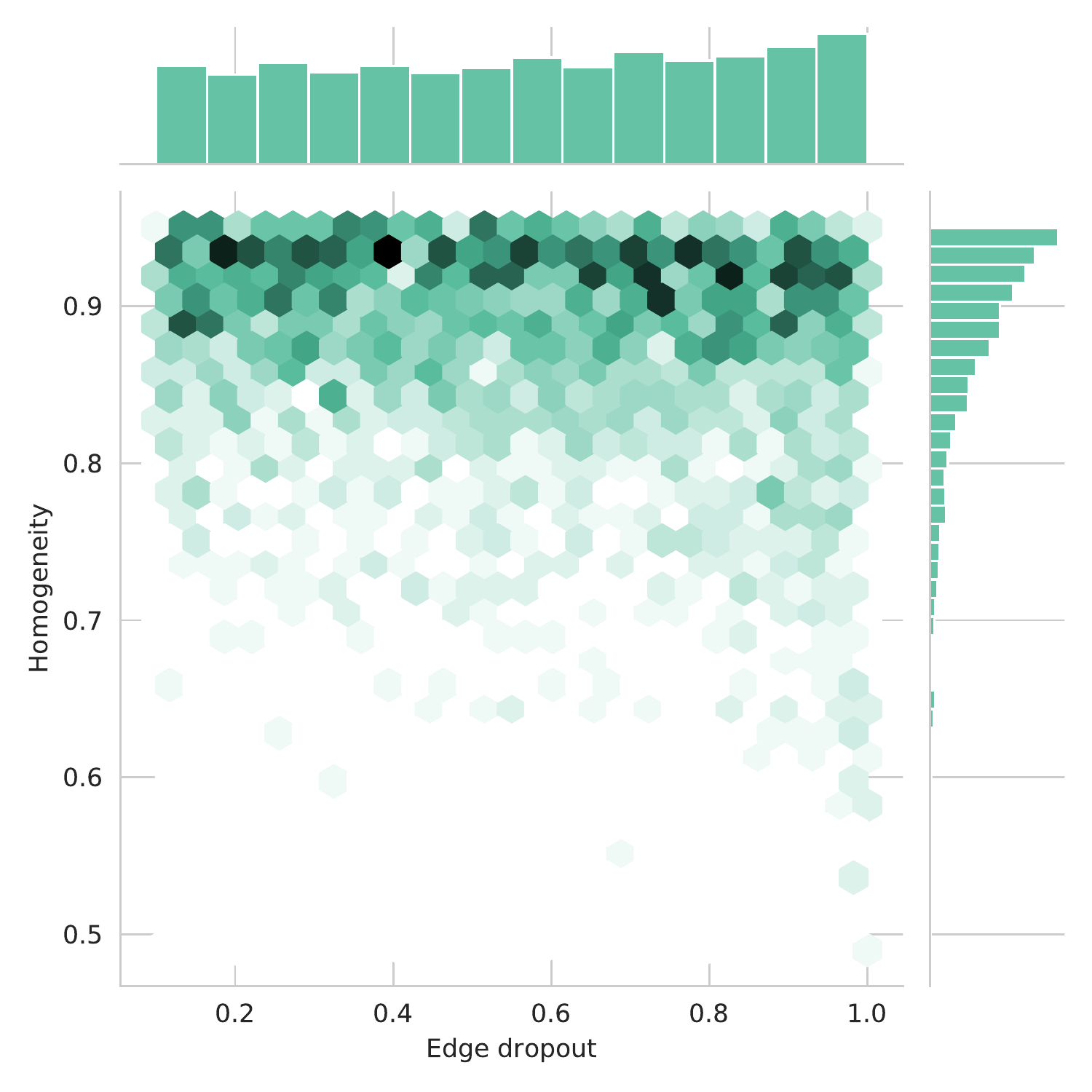}} &
\subcaptionbox{Synthetic networks: intensity $I$}{\includegraphics[width = 2.2in]{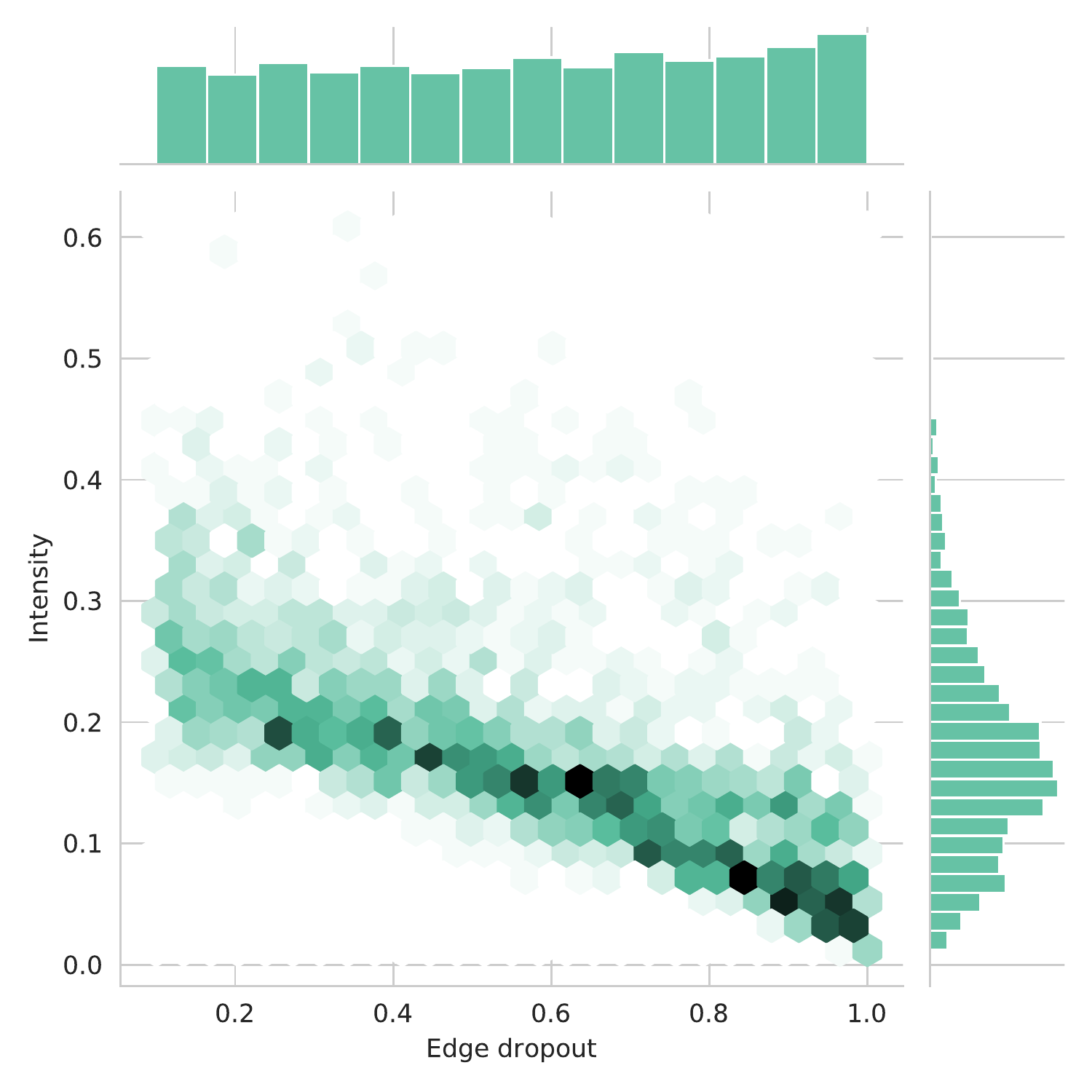}} \\
\subcaptionbox{Dependence on dropout: $H \times I$}{\includegraphics[width = 2.2in]{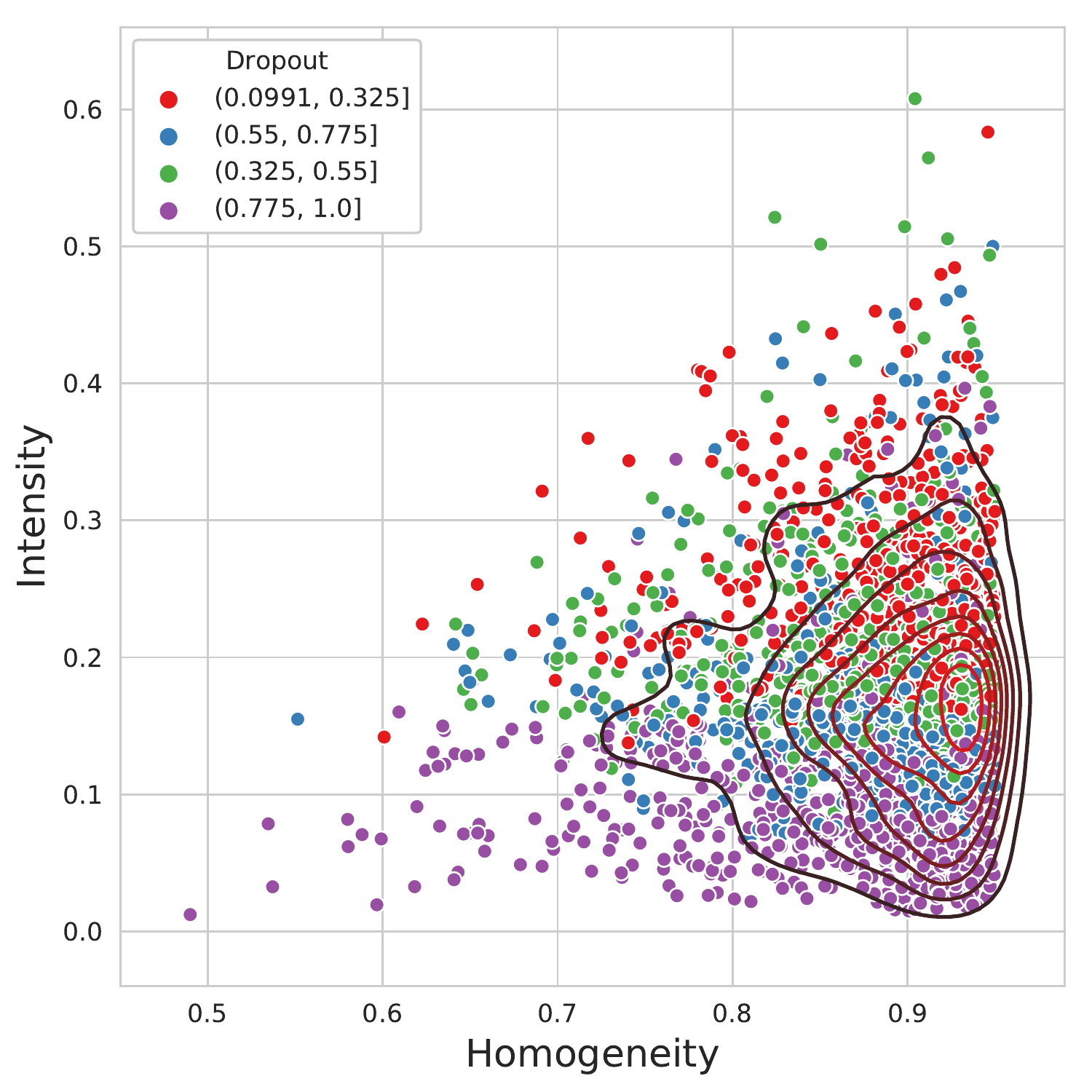}} &
\subcaptionbox{Dependence on number of layers: $H \times I$}{\includegraphics[width = 2.2in]{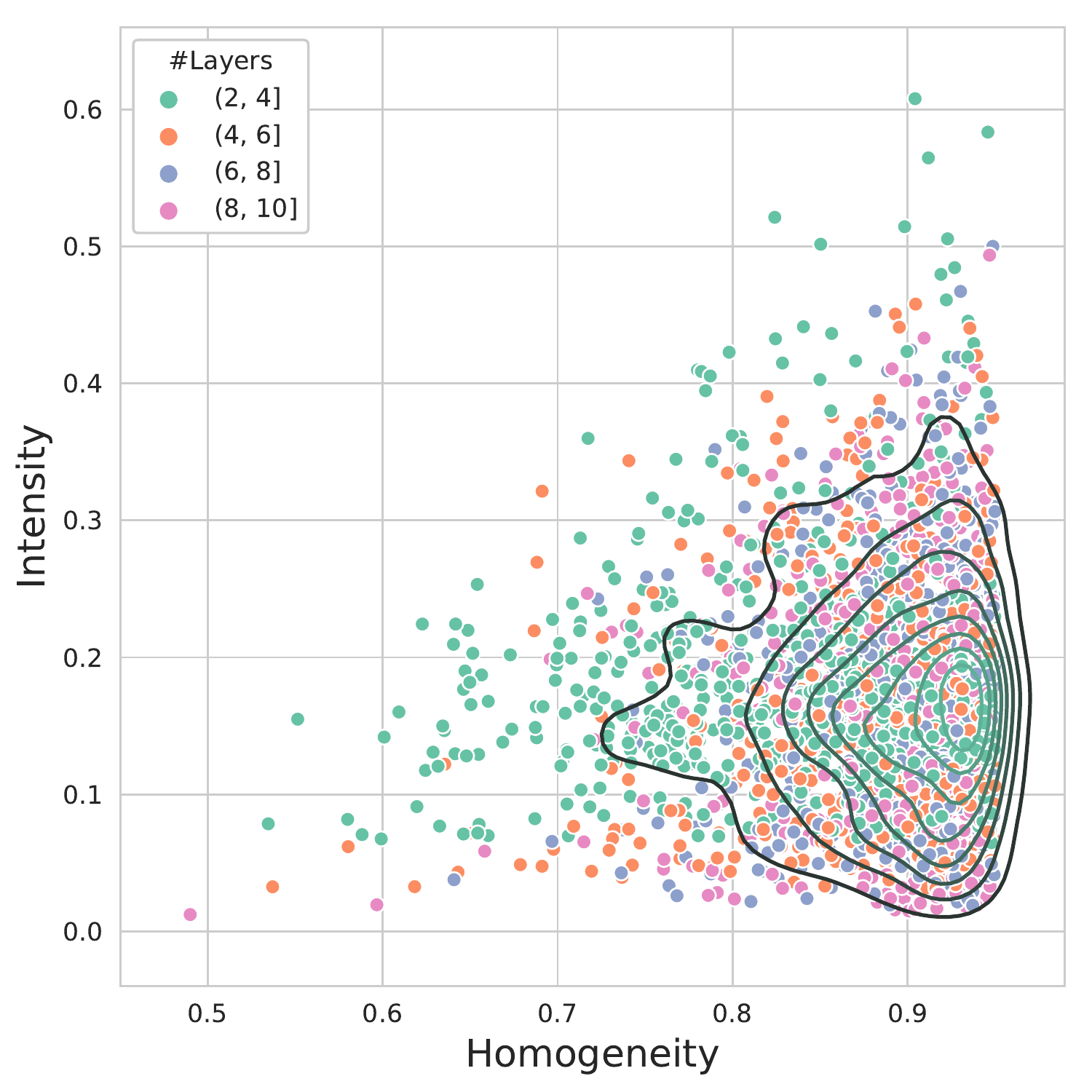}} \\
\end{tabular}

\caption{Results on 11{,}905 synthetic networks. The homogeneity (a) and intensity (b) dependence on dropout and number of layers parameter results in heavy-tailed distributions of the two measures (c, d).}
\label{fig:syn}
\end{figure}

Two main observations are apparent when studying the results on real networks. First, the difference between social and genetic (biological) multiplex networks becomes obvious when both entanglement intensity, as well as homogeneity are considered. 
We further visualize the two most apparent distributions, \textit{i.e.}, the intensity and homogeneity of social \textit{vs.} genetic networks in Figure~\ref{fig:dist}.

\begin{figure}[h]
\centering
\begin{tabular}{cc}
\subcaptionbox{Real - homogeneity}{\includegraphics[width = 2.3in]{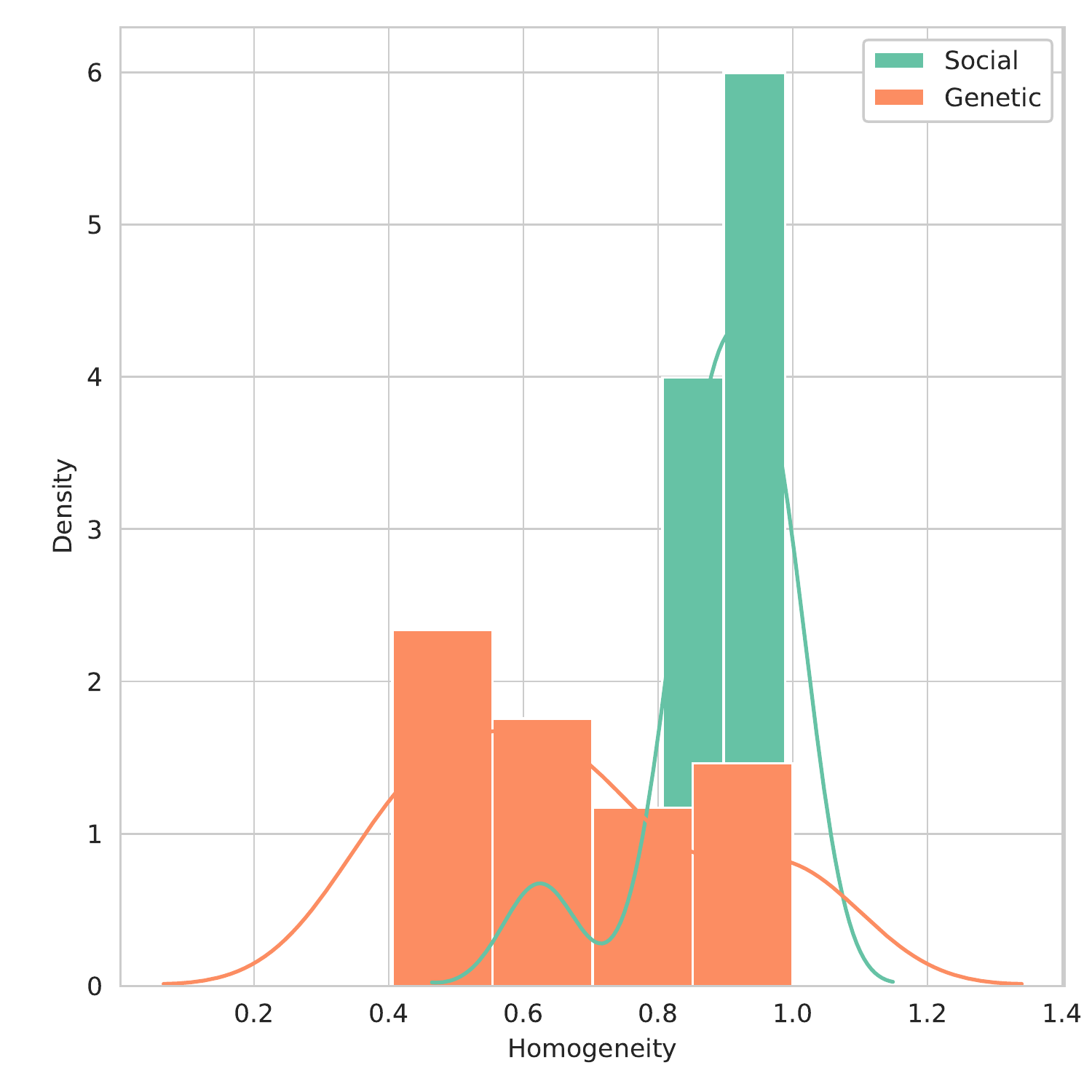}} &
\subcaptionbox{Real - intensity}{\includegraphics[width = 2.3in]{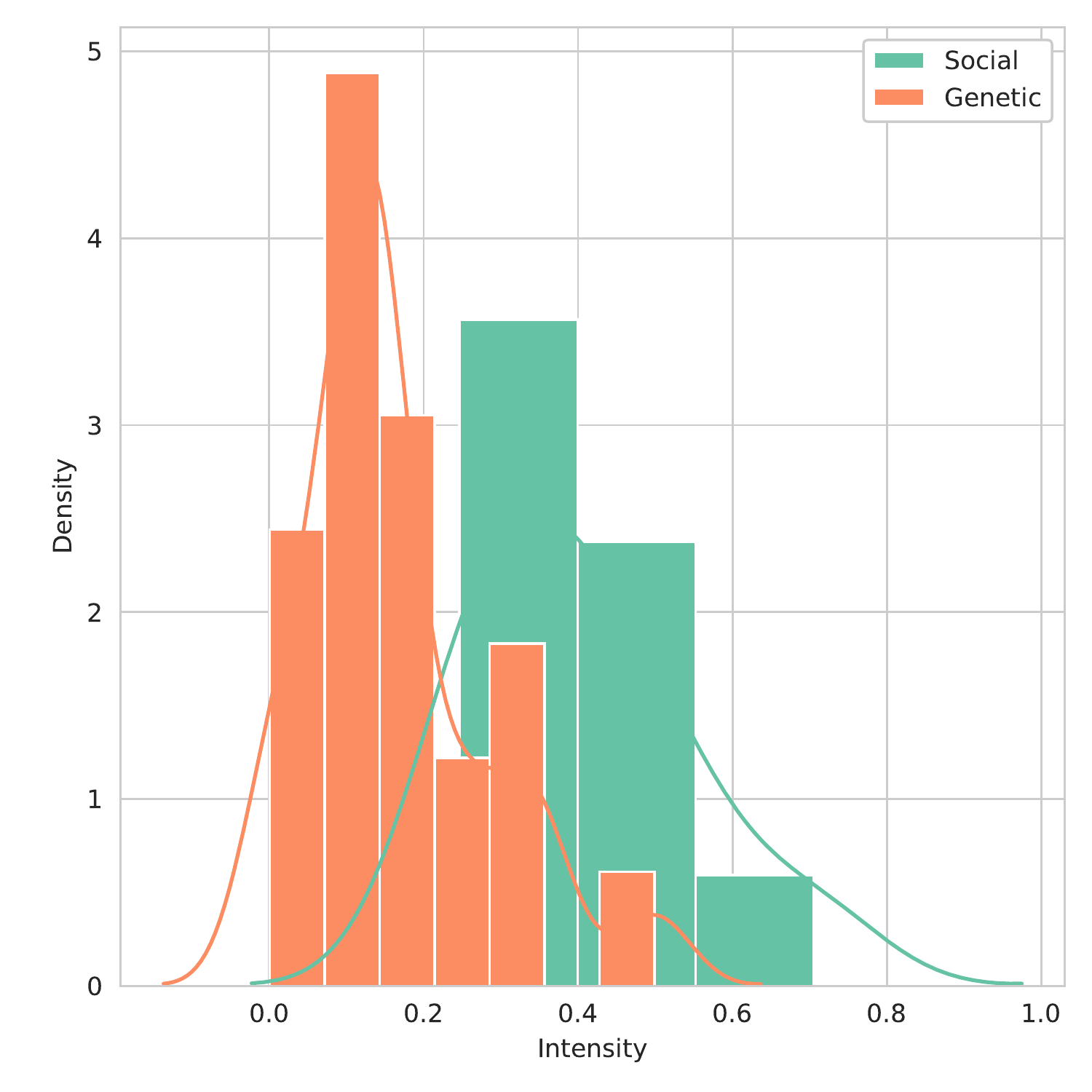}} \\
\end{tabular}
\caption{Distributions of homogeneity and intensity when genetic networks are compared to social ones.}
\label{fig:dist}
\end{figure}

The properties of synthetic networks were plotted with respect to the dropout parameter. The reader can observe apparent linear trend between the dropout (sparseness) and entanglement intensity (Figure~\ref{fig:syn} (b)). This trend indicates sparser networks are less ``intensely'' coupled. As intensity directly measures this property, this result outlines one of the \emph{desired} properties of the proposed network generator.

The reader can also observe high density of \emph{networks} in the space of high homogeneity and average or low intensity (Figure~\ref{fig:syn} (d)). This property directly reflects the sampling procedure, as the majority of the considered networks consist of edges, which co-occur in majority of layers. A similar observation can be observed in Figure~\ref{fig:syn} (a,c), where denser regions of the homogeneity/intensity space emerge when higher homogeneity is considered. Note that we also visualized the dependence of the synthetic network's properties on the dropout, as well as the number of layers --- both parameters determine a given multiplex's structure.

\section{Discussion and conclusion}
In this paper we demonstrated that two measures for assessing the relation between layers in a given multiplex network offer interesting insights when computed across a wide array of real-world networks. To our understanding, the observed relationship between the intensity and homogeneity of layer entanglement was not yet reported. We showed that real networks cluster based on their type (\textit{e.g.} biological \textit{vs.} social). Apart from experiments on real networks we also generated a large set of synthetic ones, where the analysis outlined the following properties: Intensity is directly correlated with edge dropout parameter --- the sparser the network, the lower the intensity. This result indicates the proposed generator indeed emits networks which adhere to this property. Next, we observe that large parts of the generated networks are subject to high homogeneity with various degrees of entanglement intensity.

The detailed inspection of the synthetic networks with respect to the parameters $d$ and the number of layers ($k$) reveals that the generative process is more sensitive to dropout (layered patterns of intensity emerge), than to the number of layers (uniformly distributed \textit{w.r.t.} homogeneity). This property indicates the model's properties could also be investigated theoretically, which we leave for further work. 

In addition, we may observe (from Figure~\ref{fig:real}) that our set of genetic networks tend to match networks with higher dropout, as opposed to social networks which tend to find their way in lower dropout area. This should be further investigated, but this may be related to \textit{homophily}~\cite{mcpherson2001birds,borgatti2009network}. Homophily is the implied similarity of two entities in a social network, and the property of entities to agglomerate when \textit{being similar}. If the reason of \textit{`being similar'} could be modeled as a layer of interaction, the result of a group of entities in \textit{`being similar'} would lead to the formation of a clique in this layer, hence locating social networks in low dropout areas.

The proposed work offers at least two prospects of multiplex network study which are in our belief worth exploring further. The difference between the genetic and social networks is possibly subject to very distinct topologies which emerge in individual layers. This claim can be empirically evaluated via measurement of \textit{e.g.}, graphlets, communities or other structures. Next, genetic networks are less homogeneous. Further work includes exploration of this fact, as it can be merely a property of the networks considered, empirical methodology used to obtain the networks or some other effect.

We believe that theoretical properties of the proposed network generator can also be further studies, offering potential insights into how multiplex networks behave and whether the human-made aspects are indeed representative of a given system's state.

\section{Acknowledgements}
The work of the first author was funded by the Slovenian Research Agency through a young researcher grant.
The work of other authors was supported by the Slovenian Research Agency (ARRS) core research programme \emph{Knowledge Technologies} (P2-0103) and ARRS funded research project
\emph{Semantic Data Mining for Linked Open Data} (financed under the ERC Complementary Scheme, N2-0078). We also acknowledge Dagstuhl seminar-19061~\cite{kivela2019visual} where many ideas implemented in this paper emerged.

\bibliographystyle{splncs03}
\bibliography{bibliography}

\end{document}